\documentclass[runningheads]{llncs}
\usepackage[utf8]{inputenc}
\usepackage{xcolor}
\usepackage{graphicx} 
\usepackage[hidelinks]{hyperref}
\usepackage[utf8]{inputenc}
\usepackage[small]{caption}
\usepackage{graphicx}
\usepackage[printonlyused]{acronym}
\usepackage{amsmath}
\usepackage{amsfonts}

\usepackage{amsthm}
\usepackage[ruled,vlined]{algorithm2e}
\usepackage{mathtools}
 \usepackage{amssymb}
\usepackage[caption=false, font=footnotesize]{subfig}
\usepackage{booktabs}
\usepackage{todonotes}
\usepackage{tikz}
\usetikzlibrary{automata, positioning, arrows, shapes.geometric}

\acrodef{dfa}[DFA]{Deterministic Finite Automaton}
\acrodef{scltl}[scLTL]{syntactically cosafe LTL}
\acrodef{ltl}[LTL]{Linear Temporal Logic}

\newif\ifuseboldmathops
\newif\ifuseittextabbrevs
\useboldmathopstrue   

\ifuseittextabbrevs

	\newcommand{\ie}{{\it i.e.}}

\else

	\newcommand{\ie}{i.e.}

\fi

\ifuseboldmathops
\else

\fi

\ifuseboldmathops
\else

\fi

\ifuseboldmathops
\else

\fi

\ifuseboldmathops

\else

\fi

\newcommand{\Always}{\Box \, }
\newcommand{\Eventually}{\Diamond \, }
\newcommand{\until}{\mathsf{U} \, }
\newcommand{\weakuntil}{\mathsf{W} \, }

\newcommand{\abs}[1]{\lvert#1\rvert}
\newcommand{\card}[1]{\left|#1\right|}

\newcommand{\dist}[1]{\mathcal{D}(#1)}

\newcommand{\calF}{F}

\newcommand{\calS}{S \times Q}

\newcommand{\calD}{\mathcal{D}}

\newcommand{\win}{\mathsf{Win}} 

\newtheorem{assumption}{Assumption}
\theoremstyle{plain}
\newtheorem*{problem*}{Problem}

\newcommand{\hgame}{\mathcal{H}}

\newcommand{\calA}{\mathcal{A}}
\newcommand{\game}{\mathcal{G}}

\newcommand{\act}{Act}

\newcommand{\decoy}{\mathtt{decoy}}

\newcommand{\srAct}{\mathsf{SRActs}}
\newcommand{\dswin}{\mathsf{DSWin}}

\DeclareUnicodeCharacter{1EF3}{\`y}

\begin{document}
\title{Decoy Allocation Games on Graphs with Temporal Logic Objectives}
%
%
\author{
Abhishek N. Kulkarni \inst{1} \orcidID{0000-0002-1083-8507}
\and 
Jie Fu \inst{1}\orcidID{0000-0002-4470-2827} \and
Huan Luo\inst{1}\orcidID{0000-0002-1578-9409}\and
Charles A. Kamhoua \inst{2}\orcidID{0000-0003-2169-5975}
\and 
Nandi O. Leslie
\inst{2}\orcidID{0000-0001-5888-8784}
}

\authorrunning{A. Kulkarni et al.}
%
\institute{Worceter Polytechnic Institute, Worcester MA 01609, USA\\
\email{\{ankulkarni,jfu2\}@wpi.edu, hluo12@126.com}\\
\and
U.S. Army Research Laboratory, Adelphi, MD 20783, USA \\
\email{\{charles.a.kamhoua.civ,nandi.o.leslie.ctr\}@mail.mil}}
\maketitle              
\begin{abstract}
We study a class of games, in which the adversary (attacker) is to satisfy a complex mission specified in linear temporal logic, and the defender is to prevent the adversary from achieving its goal. A deceptive defender can allocate decoys, in addition to defense actions, to create disinformation for the attacker. Thus, we focus on the problem of jointly synthesizing a decoy placement strategy and a deceptive defense strategy that maximally exploits the incomplete information the attacker about the decoy locations. We introduce a model of hypergames on graphs with temporal logic objectives to capture such adversarial interactions with asymmetric information. Using the hypergame model, we analyze the effectiveness of a given decoy placement, quantified by the set of deceptive winning states where the defender can prevent the attacker from satisfying the attack objective given its incomplete information about decoy locations. Then, we investigate how to place decoys to maximize the defender's deceptive winning region. Considering the large search space for all possible decoy allocation strategies, we incorporate the idea of compositional synthesis from formal methods  and  show that the objective function in the class of decoy allocation problem is  monotone and non-decreasing. We derive the sufficient conditions under which the objective function for the decoy allocation problem is submodular, or supermodular, respectively.  We show a  sub-optimal  allocation  can be efficiently computed by iteratively  composing the solutions of hypergames with a subset of decoys  and the solution of a hypergame given a single decoy.  We use a running example to illustrate the proposed method. 
\end{abstract}

\section{Introduction}


In security and defense applications, deception plays a key role to
mitigate the information and strategic disadvantages of the defender
against adversaries. In this paper, we investigate the design of
active defense with deception for a class of games on graphs, also known as $\omega$-regular games  \cite{gradel_automata_2002,bloem_graph_2018,chatterjee_survey_2012}. A game in this class captures the attack-defend sequential interaction in which the attacker is to complete an attack mission specified in
temporal logic \cite{manna_temporal_1992} and the defender is to mitigate attacks by selecting counter-actions and allocating decoys to create a disinformation to the attacker. We are interested in the following question: How to design the decoy allocation strategy so that the defender
can influence the attacker into taking (or not taking) certain actions that minimize the set of attacker's winning region? The winning region is defined as the set of game states from which the attacker has a strategy to successfully complete its attack mission irrespective of the defender's counter-strategy.

Games on graphs with temporal logic objectives have been studied
extensively in the synthesis of reactive programs \cite{bloem_graph_2018}. In a reactive program,  the system (player 1) is to synthesize
a program (a finite-memory strategy) to provably satisfy a desired behavior specification, no
matter which actions are taken by the uncontrollable environment (player 2). In
these games, players'  payoffs are
temporal goals and constraints, described using linear temporal logic formulas and a labeling
function. A player receives a payoff equal to one if the
\emph{labeling} over the outcome (state-sequence) of the game
satisfies its temporal logic formula. In our recent work \cite{kulkarni_deceptive_2021}, we have shown that a class of decoy-based deception can be captured by assuming that the defender has the true labels of game states but the attacker has incorrect labels. For example, a state labeled ``unsafe'' by the defender may be mislabeled as ``safe'' for the attacker. By modeling the interactions between the defender and the attacker as a hypergame, we developed the solutions of subjective rationalizable strategies for both players in this class of hypergames. The defender's subjective rationalizable strategy is by nature deceptive, as it ensures the security temporal logic specification to be satisfied by exploiting the attacker's misperception and mistakes in the attacker's subjective rationalizable strategy. We introduced deceptive winning region as the set of states (or finite game histories) from which   the defender can ensure to satisfy a security specification in this hypergame.%

However, an important problem remains: How to \emph{control the attacker's misinformation in the labeling function} so as to maximize the deceptive winning region? To restrict the freedom in crafting the disinformation, we formulate a  class of \emph{decoy-based deception game}: In this game, the defender can allocate a subset of states as hidden decoys or ``traps'', unknown to the attacker. During these interactions, the defender is to strategically select actions to lure the attacker into the traps, whereas the attacker plays rationally to satisfy her temporal logic objective given her subjective view of the interaction. In addition, the defender strategy should be \emph{stealthy}, in the sense that the attacker cannot realize a misperception exists before getting caught by one of the traps. To determine the decoy allocation, we employ the aforementioned solutions of hypergames \cite{kulkarni_deceptive_2021} to calculate the defender's deceptive sure-winning region given each individual decoys. 
  The selection of decoy locations is based on compositional synthesis \cite{filiot_antichains_2011,kulkarni_compositional_2018}, which answers, given the two deceptive sure-winning regions for  decoys allocated at two different states $s$ and $s'$, what is the deceptive sure-winning region when both states are allocated as decoys simultaneously?  We derive the sufficient conditions when the objective function for the decoy allocation problem is submodular, or supermodular, respectively.  Based on this, we can construct an under-approximation of the deceptive sure-winning regions incrementally (in polynomial time), instead of having to solve a combinatorially large number of hypergames for all possible decoy configurations.

 \paragraph{Related Work}
 Decoy allocation, also called honeypot allocation and camouflage, has been studied in recent years with applications to cyber- and physical security problems. In \cite{pibil_game_2012,kiekintveld_game-theoretic_2015}, the authors propose  a game-theoretic method to place honeypots in a network so as to maximize the probability that the attacker attacks a honeypot and not a real system. In their game formulation, the defender decides where to insert honeypots in a network,  and the attacker chooses one server to attack and receives different payoffs when attacking a real system (positive reward) or a honeypot (zero reward).  The game is imperfect information as the real systems and honeypots are indistinguishable for the attacker. By the solution of imperfect information games, the defender's honeypot placement strategy is solved to minimize the attacker's rewards. 

Security games \cite{sinha_stackelberg_2018,kiekintveld_computing_2009} are another class of important models for resource allocation in adversarial environments.  In  \cite{thakoor_cyber_2019}, the authors formulate a security game (Stackelberg game) to allocate limited decoy resources in a cybernetwork to mask network configurations from the attacker. This class of deception manipulates the adversary's perception of the payoffs and thus causes the adversary to take (or not  to take) certain actions that aid the objective of the defender.
 In \cite{durkota_optimal_2015}, the authors formulate an Markov decision process to assess the effectiveness of a fixed honeypot allocation in an attack graph, which  captures multi-stage lateral movement attacks in a cybernetwork and dependencies between vulnerabilities \cite{jha_two_2002,ouScalableApproachAttack2006}.   In \cite{anwar_honeypot_2020}, the authors analyze the honeypot allocation problem for attack graphs using normal-form games, where the defender allocates honeypots that changes the payoffs matrix of players. The optimal allocation strategy is determined using the minimax theorem. 
  The attack graph is closely related to our game on graph model, which   generalizes the attack graph to   \emph{attack-defend game graphs} \cite{jiang_optimal_2009,aslanyan_quantitative_2016} by incorporating the defender counter-actions in active defense.

There are several key distinctions between our work and the prior work. First, our work focuses on a qualitative approach to decoy allocation instead of a quantitative one, which often requires solving an optimization problem over a well-defined reward/cost function. In the qualitative approach, we represent the attacker's goal using a linear temporal logic formula, which captures rich, qualitative behavioral objectives such as reachability, safety, recurrence, persistence or a combination of these. Second, we show how to incorporate the attacker's misinformation about decoy locations into a $\omega$-regular hypergame model by representing it as labeling misperception. Hypergames \cite{bennett_hypergame_1986,sasaki_hierarchical_2016,wang_solution_1989} are a class of games with asymmetric (one-sided incomplete) information in which different players might play according to different perceptual games that capture the information and higher-order information known to that player. While the underlying idea behind our game model is similar to ``indistinguishable honeypots'' discussed in \cite{pibil_game_2012}, we are able to leverage the solution approaches for hypergames to address decoy allocation problem. Third, we solve for a stealthy strategy for the defender, which ensures that defender's actions will not inform the attacker that deceptive tactics are being used. Lastly, we borrow the idea of compositional reasoning from formal methods to find approximately optimal solutions for the decoy allocation problem for this class of hypergames.

The paper is structured as follows. In Sec.~\ref{sec:problem-formulation}, we discuss the preliminaries of attack-defend game on graph model and define the problem statement. In Sec.~\ref{sec:main-result}, we present the main results of this paper including an algorithm for the decoy allocation based on the ideas of deceptive synthesis and compositional synthesis. We employ a running example to provide intuition and illustrate the correctness as well as (near-)optimality of the proposed algorithm. Sec.~\ref{sec:conclusion} concludes the paper and discusses the future directions.

\section{Problem Formulation}
    \label{sec:problem-formulation}

\subsection{Attack-Defend Games on Graph}
    \label{subsect:game-on-graph}

In a zero-sum two-player \textit{game on graph}, player 1 (P1, pronoun `he') plays against  player 2 (P2, pronoun `she') to satisfy a given temporal logic formula.  Formally, a game on graph consists of a tuple $\game = \langle G, \varphi \rangle$, where $G$ is a \emph{game arena} modeling the dynamics of the interaction between P1 and P2, and $\varphi$ is the temporal logic specification of P1. As the game is zero-sum, the temporal logic specification of P2 is $\neg \varphi$, that is, the negation of P1's specification.

\begin{definition}[Game Arena]
    \label{def:game-arena}
    A two-player turn-based, deterministic game arena between  P1 and P2 is a tuple 
    \[ 
        G = \langle S, \act, T, AP,  L \rangle,
    \] 
    where 
    \begin{itemize}
        \item $S = S_1 \cup S_2$ is a finite set of  states partitioned into two sets $S_1$ and $S_2$. At a state  in $S_1$, P1 chooses an action. At a state  in $S_2$, P2 selects an action;
        
        \item $\act = \act_1 \cup \act_2$ is the set of actions. $\act_1$ (resp., $\act_2$) is the set of actions for P1 (resp., P2); 
        
        \item $T : (S_1 \times \act_1)\cup (S_2 \times \act_2) \rightarrow S$ is a \emph{deterministic} transition function that maps a state-action pair to a next state;
        
        
        \item $AP$ is a set of atomic propositions;
        
        \item $L: S\rightarrow 2^{AP}$ is the labeling function that maps each state $s\in S$ to a set $L(s)\subseteq {AP}$ of atomic propositions that evaluate to true at that state.
    \end{itemize}
\end{definition}

A \emph{run} in $G$ is a (finite/infinite) ordered sequence of states $\rho = (s_0, s_1, \ldots)$ such that for any $i > 0$, $s_{i} = T(s_{i-1}, a)$ for some $a \in \act$. Given the labeling function $L$, every run   $\rho$ in $G$ can be mapped to a word over an alphabet $\Sigma = 2^{AP}$ as $w = L(\rho) = L(s_0) L(s_1) \ldots$.

In this paper, we use \ac{ltl} \cite{manna_temporal_1992} to define the objectives of P1 and P2. Formally, an \ac{ltl} formula is defined as
\[
    \varphi ::= p \mid \neg \varphi \mid \varphi \land \varphi \mid \varphi \lor \varphi \mid \bigcirc \varphi \mid \varphi {\until} \varphi \mid \varphi {\weakuntil}\varphi
\]
where $p \in {AP}$ is an atomic proposition, $\neg$ (negation), $\land$ (and), and $\lor$ (or) are Boolean operators, and $\bigcirc$ (next), $\until$ (strong until) and $\weakuntil$ (weak until) are temporal operators. Formula $\bigcirc \varphi$ means that the formula $\varphi$ will be true in the next state. Formula $\varphi_1 \until \varphi_2$ means that $\varphi_2$ will be true in some future time step, and before that $\varphi_1$ holds true for every time step. Formula $\varphi_1\weakuntil \varphi_2$ means that $\varphi_1$ holds true until $\varphi_2$ is true, but does not require that $\varphi_2$ becomes true. We define two additional temporal operators: $\Eventually$ (eventually) and $\Always$ (always) as follows: $\Eventually \varphi = \top \until \varphi$ and $\Always \varphi = \neg \Eventually \neg \varphi$.

Given a word $w \in \Sigma^\omega$, let $w[i]$ be the $i$-th element in the word and $w[i\ldots]$ be the subsequence of $w$ starting from the $i$-th element. For example, for a word $w=abc$, $w[0] = a$ and $w[1\ldots] = bc$. We write $w \models \varphi$ if the word $w$ satisfies the temporal logic formula $\varphi$. 
The semantics of \ac{ltl} are defined as follows. 
\begin{itemize}
    \item{\makebox[3cm]{$w\models p$\hfill} if $p \in w[0]$;}
    
    \item{\makebox[3cm]{$w\models \neg \varphi$\hfill} if $w\not\models \varphi$;}
    
    \item{\makebox[3cm]{$w\models \varphi_1\land \varphi_2$\hfill} if $w \models \varphi_1$ and $w\models \varphi_2$;}

    \item{\makebox[3cm]{$w\models \bigcirc \varphi$\hfill} if $w[1\ldots] \models \varphi$;}

    \item{\makebox[3cm]{$w\models \varphi \until \psi$\hfill} if $\exists i \ge 0$, $w[i\ldots] \models \psi$ and $\forall 0\le j<i$, $w[j\ldots]\models \varphi$.}
    
    \item{\makebox[3cm]{$w\models \varphi \weakuntil \psi$\hfill} if  either     $w \models \varphi\until\psi$ or  $\forall 0\le j$, $w[j\ldots]\models \varphi$.}
\end{itemize}

A subclass of \ac{ltl} formula, called \ac{scltl}, does not include the weak until operator $\weakuntil$ and allows  the negation operator $\neg$ to only occur before an atomic proposition. An \ac{scltl} formula can be equivalently represented by a finite-state deterministic automaton with regular acceptance conditions, defined as follows.

\begin{definition}[Specification DFA]
Given an \ac{scltl} formula $\varphi$, its corresponding specification \ac{dfa} is a tuple \[\calA = \langle Q, \Sigma, \delta, \iota, Q_F \rangle,\] which includes a finite set $Q$ of states, a finite set $\Sigma =2^{AP}$ of symbols, a deterministic transition function $\delta: Q\times \Sigma \rightarrow Q$, a unique initial state $\iota \in Q$, and a set $Q_F \subseteq Q$ of final states.
\end{definition}

The transition function is recursively extended as $\delta(q,aw)=\delta( \delta(q, a),w )$ for given $a \in \Sigma$ and $w\in \Sigma^\ast$, where $\Sigma^\ast$ is the set of all finite words (also known as the Kleene closure of $\Sigma$). A word $w$ is \emph{accepted} by the \ac{dfa} if and only if $\delta(q, u)\in Q_F$ and $u$ is a prefix of $w$, \ie, $w=uv$ for $u\in \Sigma^\ast$ and $v\in\Sigma^\omega$, where $\Sigma^\omega$ is the set of all infinite words defined over $\Sigma$. A word is accepted by the specification \ac{dfa} $\calA$ if and only if it satisfies the \ac{ltl} formula $\varphi$.

Putting together the game arena $G$ and the \ac{scltl} objective $\varphi$ of P1, we can formally define a graphical model for the zero-sum game $\game$.

\begin{definition}[Product game]
    \label{def:product-game}
    Let $G = \langle S, \act, T, AP, L \rangle$ be a game arena and let $\calA = \langle Q, \Sigma, \delta, \iota, Q_F \rangle$ be the specification \ac{dfa} given the \ac{ltl} formula $\varphi$. Then, the product game $\game = G \otimes \calA$ is the tuple,
    \[
        \game = \langle \calS, \act, \Delta, \calF \rangle,
    \]
    where
    \begin{itemize}
        \item $S\times Q$   is a set of states partitioned into P1's states $S_1 \times Q$ and P2's states $S_2 \times Q$. 
        
        \item $\Delta: (S_1\times Q \times \act_1) \cup (S_2\times Q \times \act_2) \rightarrow \calS$ is a \textit{deterministic} transition function that maps a game state $(s, q) \in \calS$ and an action $a \in \act$ to a next state $(s',q') \in \calS$ such that $s' = T(s, a)$ and $q' = \delta(q, L(s'))$;
    
        \item $\calF = S \times Q_F$ is the set of final states in $\game$.
    \end{itemize}
\end{definition}
It is noted that we did not include an initial state in the definition of the game arena. This is because any state in $S$ can be selected to be the initial state. Let $s_0\in S$ be the initial state of the game arena, the corresponding initial state in the product game is $q_0=\delta(\iota, L(s_0))$. 
By construction, for each run $\rho=(s_0, s_1, \ldots)$ in $G$, there is a unique run $\hat \rho = (s_0, q_0), (s_1, q_1), \ldots$
in the product game, where $q_0 = \delta(\iota, L(s_0))$ for $i = 0$ and $q_i=\delta(q_{i-1}, L(s_i))$ for $i\ge 1$. The run $\rho$ satisfies the \ac{scltl} formula $\varphi$ if and only if $L(\rho)\models \varphi$ and as a result of construction, there exists $(s_i, q_i) \in \hat \rho$ for some $i \geq 0$ such that $(s_i, q_i) \in \calF$. Thus, P1's objective of satisfying an \ac{scltl} specification over the game arena $G$ is reduced to that of reaching one of the final states $\calF$ in product game $\game$. In the zero-sum game,  P2's objective of satisfying $\neg \varphi$   is reduced to preventing P1 from reaching any final states in $\calF$.

A memoryless, randomized \emph{strategy} for $i$-th player, for $i \in \{1, 2\}$, is a function $\pi_i: S_i \times Q \rightarrow \dist{\act_i}$, where $\dist{\act_i}$ is the set of discrete probability distributions over $\act_i$. 
It is noted that a memoryless strategy in a product game is a finite-memory strategy in game arena.
A strategy is deterministic if $\pi_i(\rho)$ is a Dirac delta function. 
We say that  player $i$ commits to (or follows) a strategy $\pi_i$ if and only if for a given state  $(s,q)$, if $\pi_i(s,q)$ is defined, then an action is sampled from the distribution $\pi_i(s,q)$, otherwise, player $i$ selects an action at random.
 Let $\Pi_i$ be the set of memoryless strategies of player $i$ in the product game. 

 A strategy  $\pi_1 \in \Pi_1$ is said to be sure-winning for P1 if, for every P2's strategy $\pi_2 \in \Pi_2$, P1 can ensure to reach $\calF$ in finitely many steps. A strategy $\pi_2\in \Pi_2$ is sure-winning for P2 if for every P1's strategy $\pi_1 \in \Pi_1$, P2 can ensure the game to stay in $(S \times Q) \setminus \calF$ for infinitely many steps.  The product game is known to be determined \cite{gradel_automata_2002,mcnaughton_infinite_1993}. That is, at any state $(s,q)$, only one of the players has a winning strategy and the winning strategy is memoryless.

The set of states in the product game $\game$ from which P1 (resp. P2) has a sure-winning strategy are called the \emph{sure-winning region} for P1 (resp. P2), denoted as $\win_1$ (resp. $\win_2$). 
Players' sure-winning regions can be computed by using the Alg.~\ref{alg:sure-win}
by letting $S_i\times Q$ to be $V_i$, $\act_i$ to be $A_i$, the transition function $\Delta$ and $F$ are the same components in $\game$. The interested readers are referred to Chap 2 of \cite{gradel_automata_2002} for more details.

\begin{algorithm}
\SetAlgoLined
\KwIn{A reachability game $\langle V = V_1\cup V_2, A_1\cup A_2, \Delta, F\rangle$ where $V_i$ are states where player $i$ takes an action,  $A_i$ are player $i$'s actions, $\Delta: V \times A\rightarrow V$ and P1's goal is to reach the set $F$ and P2's goal is to stay within $V\setminus F$.}
\KwOut{The winning regions $\win_1$ and $\win_2$ for P1 and P2.}
$Z_0\gets \calF$,
$Z_1\gets \emptyset$, $k\gets 0$\;
\While{$Z_{k+1}\ne Z_k$}{
 $\mathsf{Pre}_1(Z_k) \gets \{v\in V_1 \mid \exists a \in A_1 \text{ s.t. } \Delta(v, a) \in Z_k \}$\;
 $\mathsf{Pre}_2(Z_k) \gets \{v\in V_2 \mid \forall b \in A_2 \text{ s.t. } \Delta(v, b) \in  Z_k\}$\;
 $Z_{k+1} \gets Z_k \cup \mathsf{Pre}_1(Z_k) \cup \mathsf{Pre}_2(Z_k)$\; 
 $k \gets k+1$\; 
}
$\win_1\gets Z_k$, $\win_2\gets (V_1\cup V_2)\setminus \win_1$\;
\Return $\win_1,\win_2$.
 \caption{\textsc{Sure-Win}: Compute Player's Sure-Winning Regions of Zero-Sum Product Games with Reachability Objective \protect\cite{mcnaughton_infinite_1993,gradel_automata_2002}.}
 \label{alg:sure-win}
 \end{algorithm}
 
The sure-winning strategy is defined for P1 as follows: Let $Z_1,Z_2,\ldots Z_k$ be the sequence of sets generated by Alg.~\ref{alg:sure-win}, for a state $v\in (Z_i\setminus Z_{i-1})\cap V_1$, let $a$ be the action that $\Delta(v,a)\in Z_{i-1}$, then $\pi_1(v) = a$ (by construction, such an action $a$ exists). P2's sure-winning strategy is constructed as: For each $v\in \win_2$, $\pi_2(v)=a$ such that $\Delta(v,a)\in \win_2$.
Clearly, there may exist more than one sure-winning strategies for each player.

\subsection{Formulating the Decoy Allocation Problem}
    \label{subsect:decoy-allocation-problem}

We consider an interaction between the defender (P1, pronoun `he') and the attacker (P2, pronoun `she') in which the defender can use decoys to introduce incorrect information to the attacker about the game. Our goal is to investigate \textit{how to create the attacker's misinformation by allocating the decoys so as to minimize the size of the sure-winning region of the attacker.}

We now formalize the problem of decoy allocation using the game arena (Def.\ref{def:game-arena}). Let $\mathtt{decoy}$ be an atomic proposition that evaluates to true at a state if the state is equipped with a decoy. 

\begin{assumption}
    \label{assume:information-structure}
    In P2's knowledge of the game arena, no state is labeled as decoy, \ie,~ $\decoy \notin L(s)$ for all $s \in S$. 
\end{assumption}

Assumption~\ref{assume:information-structure} captures one important function of decoys---\textit{concealing fictions} \cite{heckman_bridging_2015}. The idea behind \textit{concealing fictions} is that P1 simulates the decoy states to function like a real system. As a result, P1 and P2 play with different subjective views of their interaction. With this in mind, we formalize the notion of \textit{perceptual game arena} of the players to characterize these subjective views.

\paragraph{Perceptual Game Arena.} Given that P2 does not know about the decoys, we distinguish between her view of the game arena from P1's view by introducing a different labeling function for P2. Let P1's perceptual game arena be $G^1 = G = \langle  S, \act,T, AP, L \rangle$. That is, P1 knows the ground truth. And, let P2's perceptual game arena be $G^2= \langle  S, \act,T, AP, L^2\rangle$ such that for any $s \in S$, we have $L_2(s) = L(s) \setminus \{\decoy\}$. In other words, if a state is not a decoy, then P1 and P2 share the same label for that state. If it is a decoy, then P1 knows that the proposition $\decoy$ evaluates to true at that state, but P2 does not.

\paragraph*{The Attacker and Defender Temporal Logic Objectives.} Over the perceptual game arenas $G$ and $G^2$, P1 and P2 aim to satisfy their \ac{ltl} objectives. We consider that P2's objective is specified by an \ac{scltl} formula $\varphi_2$, whose specification \ac{dfa} 
is  $\calA_2=\langle Q, \Sigma, \delta_2, \iota, Q_F \rangle$.   

Given P2's perceptual game arena $G^2$ and the specfication \ac{dfa} $\calA_2$, we can construct a perceptual product game of P2 as $\game_2 = G^2 \otimes \calA_2$. 
P1's objective is an \ac{ltl} formula  $\neg \varphi_2 \weakuntil \decoy$. That is, P1 satisfies the goal by  preventing P2 from satisfying $\varphi_2$ before reaching a decoy. However, reaching a decoy is not necessary due to the semantics of the ``weak until'' operator.

\setcounter{example}{1}
\begin{example}[Part 1]
    \label{ex:running-example-1}
    Consider a game arena as shown in Fig.~\ref{fig:arena-p2} consisting of 15 states.
    At a circle state, P1 takes an action, and at a square state, P2 takes an action. As the actions are deterministic, we use edges to indicate players' actions. For example, $(c,f), (c,g), (c,h)$ are possible actions for P2 at the state $c$.
    Over this game arena, P2 wants to satisfy an \ac{scltl} specification $\varphi_2 = \Eventually (n \lor o) \land (f \implies \Eventually n) \land (g \implies \Eventually o)$, which, in words, means that P2 must reach either the state $n$ or $o$ with the condition that whenever she visits the state $f$, she must visit $n$ and whenever she visits $g$, she must visit $o$. If she does not visit either $f$ or $g$, then she can visit either $n$ or $o$ to successfully complete her objective. The \ac{dfa} equivalent to $\varphi_2$ is shown in  Fig.~\ref{fig:dfa-varphi2}.

    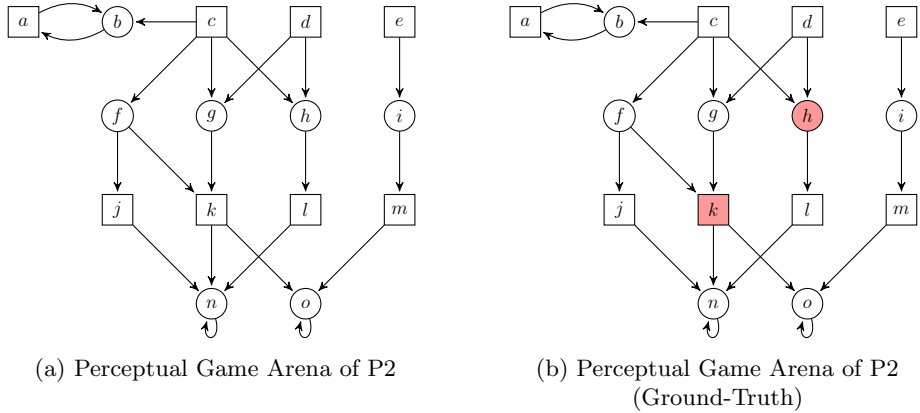
\begin{figure}[t]
    \centering
    \subfloat[Perceptual Game Arena of P2\label{fig:arena-p2}]{%
      \begin{tikzpicture}[->,>=stealth',shorten >=1pt,auto,node distance=2.5cm, scale = 0.5,transform shape]
    
      \node[state,rectangle] (a) []             {\Large $a$};
      \node[state] (b) [right of=a]             {\Large $b$};
      \node[state,rectangle] (c) [right of=b]   {\Large $c$};
      \node[state,rectangle] (d) [right of=c]   {\Large $d$};
      \node[state,rectangle] (e) [right of=d]   {\Large $e$};
      \node[state] (f) [below of=b]             {\Large $f$};
      \node[state] (g) [below of=c]             {\Large $g$};
      \node[state] (h) [below of=d]             {\Large $h$};
      \node[state] (i) [below of=e]             {\Large $i$};
      \node[state,rectangle] (j) [below of=f]   {\Large $j$};
      \node[state,rectangle] (k) [below of=g]   {\Large $k$};
      \node[state,rectangle] (l) [below of=h]   {\Large $l$};
      \node[state,rectangle] (m) [below of=i]   {\Large $m$};
      \node[state] (n) [below of=k]   {\Large $n$};
      \node[state] (o) [below of=l]   {\Large $o$};
    
      \path (a) edge[bend left]   node {} (b)
            (b) edge[bend left]   node {} (a)
            (c) edge              node {} (b)
            (c) edge              node {} (f)
            (c) edge              node {} (g)
            (c) edge              node {} (h)
            (d) edge              node {} (g)
            (d) edge              node {} (h)
            (e) edge              node {} (i)
            (f) edge              node {} (j)
            (f) edge              node {} (k)
            (g) edge              node {} (k)
            (h) edge              node {} (l)
            (i) edge              node {} (m)
            (j) edge              node {} (n)
            (k) edge              node {} (n)
            (k) edge              node {} (o)
            (l) edge              node {} (n)
            (m) edge              node {} (o)
            (n) edge [loop below]             node {} (n)
            (o) edge [loop below]              node {} (o)
            ;
            
    \end{tikzpicture}
    }\hfill
    \subfloat[\centering Perceptual Game Arena of P2 (Ground-Truth)\label{fig:arena-p1}]{%
      \begin{tikzpicture}[->,>=stealth',shorten >=1pt,auto,node distance=2.5cm, scale = 0.5,transform shape]
    
      \node[state,rectangle] (a) []                         {\Large $a$};
      \node[state] (b) [right of=a]                         {\Large $b$};
      \node[state,rectangle] (c) [right of=b]               {\Large $c$};
      \node[state,rectangle] (d) [right of=c]               {\Large $d$};
      \node[state,rectangle] (e) [right of=d]               {\Large $e$};
      \node[state] (f) [below of=b]                         {\Large $f$};
      \node[state] (g) [below of=c]                         {\Large $g$};
      \node[state,fill=red!40] (h) [below of=d]             {\Large $h$};
      \node[state] (i) [below of=e]                         {\Large $i$};
      \node[state,rectangle] (j) [below of=f]               {\Large $j$};
      \node[state,rectangle,fill=red!40] (k) [below of=g]   {\Large $k$};
      \node[state,rectangle] (l) [below of=h]               {\Large $l$};
      \node[state,rectangle] (m) [below of=i]               {\Large $m$};
      \node[state] (n) [below of=k]                         {\Large $n$};
      \node[state] (o) [below of=l]                         {\Large $o$};
    
      \path (a) edge[bend left]   node {} (b)
            (b) edge[bend left]   node {} (a)
            (c) edge              node {} (b)
            (c) edge              node {} (f)
            (c) edge              node {} (g)
            (c) edge              node {} (h)
            (d) edge              node {} (g)
            (d) edge              node {} (h)
            (e) edge              node {} (i)
            (f) edge              node {} (j)
            (f) edge              node {} (k)
            (g) edge              node {} (k)
            (h) edge              node {} (l)
            (i) edge              node {} (m)
            (j) edge              node {} (n)
            (k) edge              node {} (n)
            (k) edge              node {} (o)
            (l) edge              node {} (n)
            (m) edge              node {} (o)
            (n) edge [loop below]             node {} (n)
            (o) edge [loop below]              node {} (o)
            ;
    
    \end{tikzpicture}
    }
    
    \caption{Perceptual Game Arenas of P1 and P2 in Ex.~\ref{ex:running-example-1}.}
    \label{fig:arena}
    \end{figure}

    Suppose that P1 allocates the states $D = \{h, k\}$ as decoys. The perceptual game arenas of P1 and P2 under decoy allocation $D$ are now different. P1's perceptual game arena in Fig.~\ref{fig:arena-p1} has the same underlying graph as the perceptual game arena of P2 shown in Fig.~\ref{fig:arena-p2} but P1 has the knowledge of where the decoys are placed.  We have $\decoy \in L(h)$ and $\decoy \in L(k)$ but $\decoy \notin L(s)$ for any state $s$ except $s = h, k$. 
    Figure~\ref{fig:perceptual-games-p2} shows the perceptual product games of   P2.
    A transition $(c,0)\rightarrow (f,1)$ 
    is based on the transition $c\rightarrow f$ and $\delta_2(0, L(f)) = 1$ in the \ac{dfa} $\calA_2$ (shown in Fig.~\ref{fig:dfa-varphi2}). 
     We omit all nodes that do not have a path leading to $(n, 3)$ or $(o, 3)$. 
    
    \begin{figure}[htp]
    \centering
    \subfloat[Specification \ac{dfa} for $\varphi_2$  
        \label{fig:dfa-varphi2}]{
        \begin{tikzpicture}[->,>=stealth',shorten >=1pt,auto,node distance=3 cm, scale = 0.65,transform shape]
    \node[state,initial below] (0) [] {\large$0$};
          \node[state] (1) [above right of=0] {\large$1$};
          \node[state] (2) [below right of=0] {\large$2$};
          \node[state,accepting] (3) [right of=0, node distance=4cm] {\large$3$};

          \path 
          (0) edge[loop left]   node {$\neg (f \lor g \lor n \lor o)$} (0)
          (0) edge[right]   node {$f$} (1)
          (0) edge[]   node {$g$} (2)
          (0) edge[]   node {$\neg f \land \neg g \land (n \lor o)$} (3)
          (1) edge[loop above]   node {$\neg n$} (1)
          (2) edge[loop below]   node {$\neg o$} (2)
          (1) edge[]   node {$n$} (3)
          (2) edge[]   node {$o$} (3)
          (3) edge[loop right]   node {$\top$} (3)
          ;
        \end{tikzpicture}}
    \subfloat[Perceptual Product Game of P2\label{fig:perceptual-games-p2}]{%
      \begin{tikzpicture}[->,>=stealth',shorten >=1pt,auto,node distance=2.5cm, scale = 0.45,transform shape]

          \node[state,rectangle] (a) [] {\Large $(a, 0)$};
          \node[state] (b) [right of=a] {\Large $(b, 0)$};
          \node[state,rectangle] (c) [right of=b] {\Large $(c, 0)$};
          \node[state,rectangle] (d) [right of=c] {\Large $(d, 0)$};
          \node[state,rectangle] (e) [right of=d] {\Large $(e, 0)$};
          \node[state] (f) [below of=b] {\Large $(f, 1)$};
          \node[state] (g) [below of=c] {\Large $(g, 2)$};
          \node[state] (h) [below of=d] {\Large $(h, 0)$};
          \node[state] (i) [below of=e] {\Large $(i, 0)$};
          \node[state,rectangle] (j) [below left=1.55cm and 1.3cm of f] {\Large $(j, 1)$};
          \node[state,rectangle] (k1) [below of=f] {\Large $(k, 1)$};
          \node[state,rectangle] (k2) [below of=g] {\Large $(k, 2)$};
          \node[state,rectangle] (l) [below of=h] {\Large $(l, 0)$};
          \node[state,rectangle] (m) [below of=i] {\Large $(m, 0)$};
          \node[state,accepting] (n) [below of=k2] {\Large $(n, 3)$};
          \node[state,accepting] (o) [below of=l] {\Large $(o, 3)$};
        
          \path (a) edge[bend left]   node {} (b)
                (b) edge[bend left]   node {} (a)
                (c) edge              node {} (b)
                (c) edge              node {} (f)
                (c) edge              node {} (g)
                (c) edge              node {} (h)
                (d) edge              node {} (g)
                (d) edge              node {} (h)
                (e) edge              node {} (i)
                (f) edge              node {} (j)
                (f) edge              node {} (k1)
                (g) edge              node {} (k2)
                (h) edge              node {} (l)
                (i) edge              node {} (m)
                (j) edge              node {} (n)
                (k1) edge              node {} (n)
                (k2) edge              node {} (o)
                (l) edge              node {} (n)
                (m) edge              node {} (o)
                (n) edge [loop below]             node {} (n)
            (o) edge [loop below]              node {} (o)
            ;
        \end{tikzpicture}
     }
    
    \caption{P2's specification \ac{dfa} and the perceptual product game in Ex.~\ref{ex:running-example-1}.}
    \label{fig:perceptual-games}

    \end{figure}
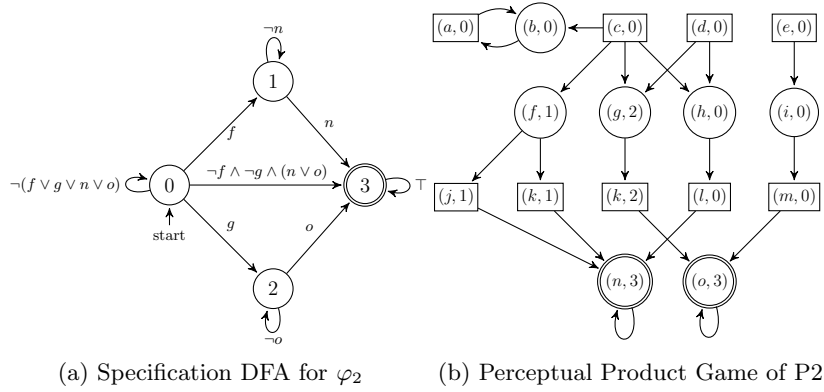
    
\end{example}

We now formalize our problem statement.
\begin{problem}
    \label{prob:decoy-placement}
    Given a set of  $k $ decoys and a set $\calD\subseteq S$ of states at which decoys can be placed,  identify the decoy locations $D\subseteq \calD$ with $\card{D} \leq k$ such that by letting $\decoy \in L(s)$ for each $s\in D$,  the number of states in the product game $\game_1$  from which P1 has a strategy to satisfy the security specification $\varphi_1$ is maximized, given that P2 may choose any counter-strategy that she considers rational in her perceptual game, $\game_2$.
\end{problem}
The objective of P1 is intuitively understood as to maximize the set of system states protected by the defense strategy.

\section{Main Result}
    \label{sec:main-result}

Our proposed solution to Problem~\ref{prob:decoy-placement} is based upon two key ideas from formal methods and hypergame theory, namely (a) deceptive synthesis, and (b) compositional synthesis. In Sec.~\ref{subsect:hgame}, we introduce  deceptive synthesis to construct a strategy for P1 to deceive P2 into reaching a pre-defined decoy set in finitely many steps by exploiting the  incomplete information of P2. The strategy is called \emph{deceptive sure-winning strategy} and depends on the chosen set of decoys. Then, in Sec.~\ref{subsect:compsynth}, we introduce a compositional synthesis approach to identify an approximately optimal allocation of decoys. 

\subsection{Deceptive Synthesis: Hypergames on Graphs}
    \label{subsect:hgame}

Consider a set $D\subseteq S$ of states are allocated with decoys, unknown to P2. In such an interaction, as seen in Sec.~\ref{subsect:decoy-allocation-problem}, the players have different perceptual game arenas that share the same set of states, actions, and transitions but different labeling functions. 
We introduce  a model of hypergame on graph to integrate the games $\game_1$ of P1 and $\game_2$ of P2 into a single graphical model. 

\begin{definition}[Hypergame on Graph  (modified from \protect{\cite{kulkarni_deceptive_2021}}\footnote{
Def.~\ref{def:hgame} is a simplified version of \cite[Def. 6]{kulkarni_deceptive_2021}, which considers the general case when P1 and P2's objectives are both general \ac{scltl} formulas, not necessarily in the current form.})]
    \label{def:hgame}
    Given the perceptual game arenas $G = \langle S, \act, T, AP, L \rangle$ and $G^2 = \langle S, \act, T, AP, L^2 \rangle$, and P2's specification \ac{dfa} $\calA_2 = \langle Q, 2^{AP}, \delta_2, \iota, Q_F\rangle$,  let $D \subsetneq S$ be a set of states such that $\decoy \in L(s)$.  The \emph{hypergame on graph} given the players' objectives $\neg \varphi_2\weakuntil \decoy$ for P1 and $\varphi_2$ for P2 is a transition system
    \[
        \hgame_D = \langle S \times Q, \act, \Delta, F_D, F_2 \rangle, 
    \]
    where
    \begin{itemize}
        \item $ S \times Q $ is the set of states;
        
        \item $\Delta: (S_1\times Q \times \act_1) \cup (S_2\times Q \times \act_2) \rightarrow \calS$ is a \emph{deterministic} transition function such that $\Delta((s, q), a) = (s', q')$ if and only if $s' = T(s, a)$ and $q' = \delta_2(q, L_2(s'))$;
        
       \item $F_D = \{(s,q) \mid \decoy \in L(s)\}$ is the set of states which P1 must reach in order to satisfy $\neg \varphi_2 \weakuntil \decoy$; 
       
        \item $F_2 = \{(s, q) \mid q\in Q_F\}$ is the set of final states which P2 must reach in order to satisfy $\varphi_2$.
    \end{itemize}
\end{definition}

It is noted that the sets of states, actions,  transitions, and P2's final states $F_2$ in $\hgame_D$ are defined exactly as these components in P2's perceptual product game $\game_2$ (see Def.~\ref{def:product-game}). The additional set $F_D$ is introduced to represent P1's objective. 

Let us denote the sure-winning region of player $i$ in player $j$'s perceptual game $\game_j$ by $\win_i^j$. The attacker's perceptual winning regions can be solved with the attacker's reachability  game using Alg.~\ref{alg:sure-win} by letting $V_1\coloneqq S_2\times Q$, $V_2\coloneqq S_1\times Q$, $A_1\coloneqq \act_2$, $A_2\coloneqq \act_1$, $\Delta$ is the same as in $\hgame_D$, and $F\coloneqq F_2$. The following observations are noted:

\begin{itemize}
    \item For every state $(s, q) \in \win_1^2$ (P1's sure-winning region perceived by P2), P1 can ensure to satisfy $\neg \varphi_2$ no matter which strategy P2 uses. Decoys are not  needed for states within $\win_1^2$.
    
    \item For every state  $(s, q) \in \win_2^2$ (P2's sure-winning region perceived by P2), P2 can ensure satisfying $\varphi_2$ when no decoy is used. However, when decoys are introduced, P1 can exploit P2's lack of knowledge about the decoys and lure P2 into reaching decoys before P2 is able to satisfy $\varphi_2$.
\end{itemize}

It is  known \cite{baier2008principles} that we can rewrite $\neg \varphi_2 \weakuntil \decoy$ using the temporal operators: $\until$ (until) and $\Always$ (always), as $(\neg \varphi_2 \until \decoy) \lor \Always \neg \varphi_2$, where $\Always \varphi = \neg \Eventually \neg \varphi$. 
When the game state is within P2's perceptual winning region $\win_2^2$, then P1 does not have a strategy to ensure $\Always \neg \varphi_2$ (reads ``always $\varphi_2$ is false'') and can only satisfy his specification by enforcing P2 to visit a decoy. The following Lemma formalizes this statement.

\begin{lemma}  
    \label{lma:P1-Spec-Is-Cosafe}
     For any state $(s,q)\in \win_2^2 $, any strategy $\pi_1$ of P1 that satisfies $\neg \varphi_2\weakuntil \decoy$ also satisfies $\neg \varphi_2\until \decoy$. 
 \end{lemma}
 We omit the proof noting that it follows from the definition of weak until and the property of winning region.

Thus, when we focus our attention on the region $\win_2^2$, P1's objective is equivalently $\neg \varphi_2\until \decoy$. Before addressing the decoy allocation problem, we must answer: From which states in $\win_2^2$, P1 can ensure to satisfy $\neg \varphi_2\until \decoy$ by exploiting P2's lack of knowledge about the decoy states, \ie, $F_D$?

To answer this question, we formulate a deceptive game for P1. We first   restrict P1's actions to those considered rational for P2 in her perceptual game. At the same time, P2's irrational actions are removed as P1 knows a rational P2 will not use these actions.
As the rational actions are based on P2's subjective view of the game, we formalize this notion of rationality using the concept of \emph{subjective rationalizability} from game theory (we refer the interested readers to \cite{kulkarni_deceptive_2021} for rigorous treatment).

\begin{definition}[Subjectively Rationalizable Actions in $\game_2$]
    \label{def:subjectively-rationalizable-actions}
    Given P2's perceptual product game   $\game_2 = \langle \calS, \act, \Delta, F_2 \rangle$, a player $i$'s action $a \in \act_i$ is said to be \emph{subjectively rationalizable} at his/her winning state $(s, q) \in \win_i^2$ in $\game_2$ if and only if $\Delta((s, q), a) \in \win_i^2$. At player $i$'s losing state $(s, q) \notin \win_i^2$, any action of player $i$ is assumed to be subjectively rationalizable for player $i$. 
\end{definition}

Based on Def.~\ref{def:subjectively-rationalizable-actions}, we define the set of subjectively rationalizable actions of player $i$ at a state $(s, q) \in S \times Q$ as follows:
\begin{align}
    \srAct_i^{2}(s, q) = & \{a \in \act_i \mid (s, q) \in \win_i^2 \text{ and } \Delta((s, q), a) \in \win_i^2\} ~\cup \nonumber \\
    & \{a \in \act_i \mid (s, q) \notin \win_i^2 \text{ and } \Delta((s, q), a)  \text{ is defined} \}
\end{align}

\begin{assumption}
    \label{assume:P2change-information-action}  
    Subjective rationalizability is a common knowledge between P1 and P2. 
\end{assumption}

Assumption~\ref{assume:P2change-information-action} means that both players know that their opponent is subjectively rational  and that the opponent is aware of this fact. Thus,  P2 would become aware of her misperception in the game arena, when P1 uses an action which is not subjectively rationalizable in P2's perceptual game, $\game_2$. We can refine the hypergame on graph $\hgame_D$ to eliminate: 1) states that do not require decoys: This is the set $\win_1^2$ from which P1 has a sure-winning strategy for $\neg \varphi_2$; 2) actions that contradict P2's perception. After this elimination, we obtain a deceptive reachability game for P1, for synthesizing P1's deceptive strategy.

\begin{definition}[P1's deceptive reachability game]
    \label{def:hgame-labeling-misperception}
    Given the hypergame on graph $\hgame_D = \langle S \times Q, \act, \Delta, F_D, F_2 \rangle$, P1's deceptive reachability game is 
    \[
        \widehat\hgame_D  = \langle \win_2^2, \act, \widehat \Delta, F_D \rangle, 
    \]
    where 
    \begin{itemize}
        \item $\win_2^2$ is a set of P2's perceptual winning states, and game state space for P1's deceptive reachability game.
        
        \item $\widehat \Delta: S\times Q \times \act  \rightarrow S\times Q$ is a \emph{deterministic} transition function such that 
        \begin{itemize}
            \item if $(s,q) \notin F_2$ then $\widehat \Delta((s, q ), a) = \Delta((s,q),a)$ whenever $s \in S_i$ and $a \in \srAct_i^2(s,q)$ for $i = 1, 2$. Otherwise, $\widehat \Delta((s, q ), a)$ is undefined.
            
            \item if $(s,q)\in F_2$, then for any action $a\in \act$, $\widehat \Delta((s, q ), a) = (s,q)$. That is, the set $F_2$ are modified into sink states. 
        \end{itemize}
        
        \item $F_D$ is the  set  of  states that P1  aims to  reach.
    \end{itemize}
    
\end{definition}

\begin{lemma}
For a given state $(s,q)$, if P1 has a sure-winning strategy in $\widehat \hgame_D$ starting from $(s,q)$, then P1 can ensure to satisfy $\neg \varphi_2\until \decoy$ by following this sure-winning strategy in $\widehat \hgame_D$.
\end{lemma}
\begin{proof}
A path satisfies $\neg\varphi_2\until \decoy$ if it reaches $F_D$ and before reaching $F_D$, it does not visit any state in $F_2$. 
By construction of $\widehat \hgame_D$, if  any path reaches $F_D$, it must not have visited $F_2$ because if $F_2$ is reached prior to $F_D$, then the game stays in the sink state and will never reach $F_D$. Thus, P1's sure-winning strategy that ensures a path to reach $F_D$ alone  satisfies $\neg \varphi_2\until \decoy$.
\end{proof}

Formally, P1's sure-winning strategy $\pi_1$ in the deceptive reachability game is said to be \emph{deceptively sure winning}. 
A state from which P1 has a deceptive sure-winning strategy is called a \emph{deceptively sure-winning state}. The set of all deceptively sure-winning states of P1 in $\widehat \hgame_D$ is called P1's \emph{deceptive sure-winning region}. The deceptive sure-winning region for P1 can be computed by using Alg.~\ref{alg:sure-win} with $\widehat\hgame_D$ by letting $V_1\coloneqq (S_1\times Q) \cap \win_2^2$, 
$V_2\coloneqq (S_2\times Q) \cap \win_2^2$, 
$\Delta \coloneqq \widehat\Delta$, and $F\coloneqq F_D$ (see the description of terms in Alg.~\ref{alg:sure-win}). We denote the deceptive sure-winning region for P1 as $\dswin_D$.

It is noted that the deception is induced by the set $F_D$ which is hidden from P2, and the fact that during the interaction, P1 does not choose any action that contradicts P2's misperception.  Additionally, we note that deceptive sure-winning region is not defined for P2, as she is unaware of her lack of information until a decoy is reached.

We now continue with the running example to illustrate the hypergame and P1's deceptive reachability game.
\setcounter{example}{1}
\begin{example}[Part 3]
    \label{ex:running-example-3}
    From Def.~\ref{def:hgame}, we note that the hypergame on graph $\hgame_D$ shares the same underlying graph as P2's perceptual game, $\game_2$. That is, in our example, $\hgame_D$ would have the same graph as Fig.~\ref{fig:perceptual-games-p2} but has the states $(h, 0), (k, 1)$ and $(k, 2)$ labeled as the sink states (shown in red). Now, let us understand the construction of $\widehat\hgame_D$ from $\hgame_D$. We start by computing $\win_2^2$ using Alg.~\ref{alg:sure-win} over the model $\game_2$ by letting $V_1 \coloneqq S_2 \times Q, V_2 \coloneqq S_1 \times Q, \Delta \coloneqq \Delta$ and $F \coloneqq F_2$. This results in $\win_2^2$ to include all states except $(a, 0), (b, 0)$. 
    Intuitively, at the state $(b,0)$, P1 can always choose the transition $b\rightarrow a$ to reach $(a,0)$ and keep the game state within $\{(a,0),(b,0)\}$. 
    Consequently, any action that leads to $(a, 0), (b, 0)$ is \textbf{not} subjectively rationalizable for P2 and thereby removed. Additionally, the states $(a, 0)$ and $(b, 0)$ are also removed from $\hgame_D$ to get $\widehat\hgame_D$, which is shown in Fig.~\ref{fig:widehat-hgame}.

    \begin{figure}
        \centering
        \begin{tikzpicture}[->,>=stealth',shorten >=1pt,auto,node distance=2.5cm, scale = 0.5,transform shape]

      \node[state,rectangle] (c) [right of=b] {\Large $(c, 0)$};
      \node[state,rectangle] (d) [right of=c] {\Large $(d, 0)$};
      \node[state,rectangle] (e) [right of=d] {\Large $(e, 0)$};
      \node[state] (f) [below of=b] {\Large $(f, 1)$};
      \node[state] (g) [below of=c] {\Large $(g, 2)$};
      \node[state,fill=red!40] (h) [below of=d] {\Large $(h, 0)$};
      \node[state] (i) [below of=e] {\Large $(i, 0)$};
      \node[state,rectangle] (j) [below left=2.525cm and 2.25cm of f] {\Large $(j, 1)$};
      \node[state,rectangle,fill=red!40] (k1) [below of=f] {\Large $(k, 1)$};
      \node[state,rectangle,fill=red!40] (k2) [below of=g] {\Large $(k, 2)$};
      \node[state,rectangle] (l) [below of=h] {\Large $(l, 0)$};
      \node[state,rectangle] (m) [below of=i] {\Large $(m, 0)$};
      \node[state,accepting] (n) [below of=k2] {\Large $(n, 3)$};
      \node[state,accepting] (o) [below of=l] {\Large $(o, 3)$};
    
      \path 
            (c) edge              node {} (f)
            (c) edge              node {} (g)
            (c) edge              node {} (h)
            (d) edge              node {} (g)
            (d) edge              node {} (h)
            (e) edge              node {} (i)
            (f) edge              node {} (j)
            (f) edge              node {} (k1)
            (g) edge              node {} (k2)
            (h) edge              node {} (l)
            (i) edge              node {} (m)
            (j) edge              node {} (n)
            (k1) edge              node {} (n)
            (k2) edge              node {} (o)
            (l) edge              node {} (n)
            (m) edge              node {} (o)
            (n) edge [loop below]              node {} (n)
            (o) edge [loop below]              node {} (o)
            ;
    \end{tikzpicture}
        \caption{P1's deceptive reachability game.}
        \label{fig:widehat-hgame}
    \end{figure}
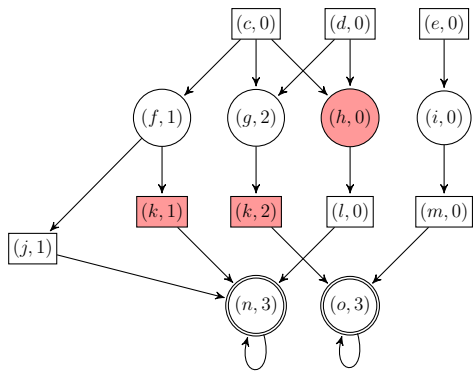
\end{example}

\subsection{Compositional Synthesis for Decoy Allocation}
    \label{subsect:compsynth}

Given a subset $\calD\subseteq S$ of states that can be allocated as decoys, for every different choice of decoy allocation $D\subseteq \calD$ we have a different hypergame, $\widehat\hgame_D$. In this context, solving Problem~\ref{prob:decoy-placement} is equivalent to identifying one hypergame that has the largest  deceptive sure-winning region $\abs{\dswin_D}$ for P1. A na\"ive approach to solve this problem would be to compute $\dswin_D$ for each $D \subseteq \calD$ and then select a set $D$ for which $\abs{\dswin_D}$ is the largest. However, this approach is not scalable because the number of subsets increases combinatorially with the size of game. To address this issue, we introduce a compositional approach to decoy allocation in which we show that when certain conditions hold, the decoy allocation problem can be formulated as a sub or supermodular optimization problem. We propose an algorithm to approximate the optimal decoy allocation.

\begin{proposition}
    \label{prop:or-composition}
    Let $\dswin_{\{s_1\}}$ and $\dswin_{\{s_2\}}  $ be P1's deceptive sure-winning regions in the hypergames $\widehat\hgame_{\{s_1\}}$ and $\widehat\hgame_{\{s_2\}}$ respectively. Then, P1's deceptive sure-winning region $\dswin_{\{s_1, s_2\}}$ in the reachability game $\widehat\hgame_{\{s_1, s_2\}}$ is equal to the sure-winning region for P1 in the following zero-sum, reachability game:
    \[\widehat \hgame_{\{s_1, s_2\}} = \langle \win_2^2, \act, \widehat\Delta, \dswin_{\{s_1\}} \cup \dswin_{\{s_2\}} \rangle,\] 
    where P1's goal is to reach the target set $ \dswin_{\{s_1\}} \cup \dswin_{\{s_2\}}$
and P2's goal is to prevent P1 from reaching the target set.   

\end{proposition}

\begin{proof}
First, it is noted that  all the three deceptive reachability games: $\widehat\hgame_{\{s_1\}}$,  $\widehat\hgame_{\{s_2\}}$ and $\widehat\hgame_{\{s_1, s_2\}}$, share the same underlying graphs but different reachability objectives for P1: $F_{\{s_1\}}, F_{\{s_2\}}$, and $F_{\{s_1,s_2\}}$. In addition, $F_{\{s_1\}}\cup F_{\{s_2\}} = F_{\{s_1,s_2\}}$.
  By definition of sure-winning regions, from every state $(s, q) \in \dswin_{\{s_i\}}$ for $i = 1, 2$, there exists a deceptive sure-winning strategy $\pi_{\{s_i\}}^\ast$ for P1 to ensure  $F_{\{s_i\}}$ is reached in finitely many steps, for any subjectively rationalizable counter-strategy of P2.
  
In $\widehat \hgame_{\{s_1, s_2\}}$, let $W^\ast \subseteq \win_2^2$ be the sure-winning region for P1  and $\pi^\ast$ be the sure-winning strategy of P1.
    From a state $(s,q)$ in $W^\ast$, P1 can ensure to reach a state, say $(s',q')\in \dswin_{\{s_1\}}\cup \dswin_{\{s_2\}}$ by following $\pi^\ast$. Upon reaching a state $(s',q')$, P1 can ensure to reach a state in either $F_{\{s_1\}}$ or $F_{\{s_2\}}$---that is, P1 can ensure to reach a state in $F_{\{s_1,s_2\}} $.
    Hence, a sure-winning  state $(s,q)$ in the above reachability game is deceptive sure-winning
 in $\widehat \hgame_{\{s_1,s_2\}}$ in which $F_{\{s_1,s_2\}}$ is P1's reachability objective. The deceptive sure-winning strategy
  is sequentially composed of strategies $\pi^\ast$, $\pi_{\{s_1\}}^\ast$, and $\pi_{\{s_1\}}^\ast$ as follows: From a state $(s,q)\in W^\ast$, P1 uses $\pi^\ast$ until a state in $\dswin_{\{s_1\}}\cup \dswin_{\{s_2\}}$ is reached. If $\dswin_{\{s_1\}} \setminus \dswin_{\{s_2\}}$ is reached, P1 uses the  sure-winning strategy $\pi^\ast_{\{s_1\}}$; If $\dswin_{\{s_2\}}\setminus \dswin_{\{s_1\}}$ is reached, P1 uses the sure-winning stratgy $\pi^\ast_{\{s_2\}}$; if $\dswin_{\{s_1\}} \cap \dswin_{\{s_2\}}$, P1 selects one of $\pi^\ast_{\{s_1\}}$ and $\pi^\ast_{\{s_2\}}$ arbitrarily.
\end{proof}

Prop.~\ref{prop:or-composition} provides us a way for composing the deceptive sure-winning regions of two deceptive reachability games   $\widehat\hgame_{s_1}$ and $\widehat\hgame_{s_2}$ to compute the deceptive sure-winning region in the deceptive reachability game $\widehat\hgame_{\{s_1, s_2\}}$ where both $s_1$ and $s_2$ are allocated as decoys. A more general result can be obtained by applying Prop.~\ref{prop:or-composition} repeatedly.

\begin{corollary}
    \label{cor:or-composition}
    Given $\dswin_{D} $  and $\dswin_{\{s\}}$ as P1's deceptive sure-winning regions in hypergames $\widehat\hgame_{D}$ and $\widehat\hgame_{\{s\}}$ respectively,
    P1's deceptive sure-winning region $\dswin_{D \cup \{s\}}$
     in the deceptive reachability game $\widehat\hgame_{D \cup \{s\}}$   
    equals the sure-winning region for P1 in the following zero-sum, reachability game:
    \[
    \langle \win_2^2, \act, \widehat\Delta, \dswin_{D} \cup \dswin_{\{s\}} \rangle
    \] 
    where P1's goal is to reach the target set $ \dswin_{D} \cup \dswin_{\{s\}}$
and P2's goal is to prevent P1 from reaching the target set.   
\end{corollary}

\begin{corollary}
    \label{cor:subset}
    Given a set $D \subseteq \calD$ and a state $s \in \calD$, we have 
    \[
        \dswin_{D} \cup \dswin_{\{s\}} \subseteq \dswin_{D \cup \{s\}}
    \]
\end{corollary}

Corollary~\ref{cor:subset} follows immediately from Proposition~\ref{prop:or-composition} and  Alg.~\ref{alg:sure-win}. To see this, consider a P1 state $s \in \calD$ which is neither in $\dswin_D$ nor in $\dswin_{\{s\}}$ but has exactly two transitions: one leading to $s$ and another leading to a state in $\dswin_D$. Clearly, the new state will be added to $\dswin_{D \cup \{s\}}$. Thus, if we consider the size of $\dswin_D$ to be a measure of effectiveness of allocating the states in $D \subseteq \calD$ as decoys, then Corollary~\ref{cor:subset} states that the effectiveness of adding a new state to a set of decoys is greater than or equal to the sum of their individual effectiveness.

\setcounter{example}{1}
\begin{example}[Part 4]
    \label{ex:running-example-4}
    Given the underlying graph of P1's reachability game $\widehat\hgame_D$ from Fig.~\ref{fig:widehat-hgame}, let us observe the effect of choosing different $D$ on P1's deceptive sure-winning region, $\dswin_D$. Letting $k = 2$,  Fig.~\ref{fig:decoy-placement-choice} shows the $\dswin_D$ for $D = \{h, k\}$ (Fig.~\ref{fig:decoy-placement-choice-1}) and $D = \{l, m\}$ (Fig.~\ref{fig:decoy-placement-choice-2}). In the figure, the colored states represent P1's deceptive sure-winning region, $\dswin_D$. The states in $F_D$ are colored red and the states from which P1 has deceptive sure-winning strategy to reach a state in $F_D$ are colored blue. For instance, for $D = \{h, k\}$, a P1 state $(f, 1)$ is included in $F_{\{h, k\}}$ because there exists an action for P1 that leads to $(k, 1)$, which is in $F_{\{h, k\}}$. Similarly, a P2 state $(d, 0)$ is included in $\dswin_{\{h, k\}}$ because both the outgoing transitions from $(d, 0)$ lead to a deceptively sure-winning state. We also notice that the states $(c, 0)$ and $(d, 0)$ from $\dswin_{\{h,k\}}$ are \emph{not} included in either $\dswin_{\{h\}} = \{(h, 0)\}$ or $\dswin_{\{k\}} = \{(k, 1), (k, 2), (f, 1), (g, 2)\}$ because both the states have at least one transition that does not lead to deceptive sure-winning state. For instance, the transition $(d, 0) \rightarrow (g, 2)$ prevents the state $(d, 0)$ to be added to $\dswin_{h}$.

    \begin{figure}[htp]
    \centering
    \subfloat[\centering Deceptive sure-winning region of P1 when $D = \{h, k\}$\label{fig:decoy-placement-choice-1}]{%
      \begin{tikzpicture}[->,>=stealth',shorten >=1pt,auto,node distance=2.5cm, scale = 0.45,transform shape]

          \node[state,rectangle,fill=blue!10] (c) [right of=b] {\Large $(c, 0)$};
          \node[state,rectangle,fill=blue!10] (d) [right of=c] {\Large $(d, 0)$};
          \node[state,rectangle] (e) [right of=d] {\Large $(e, 0)$};
          \node[state,fill=blue!10] (f) [below of=b] {\Large $(f, 1)$};
          \node[state,fill=blue!10] (g) [below of=c] {\Large $(g, 2)$};
          \node[state,fill=red!40] (h) [below of=d] {\Large $(h, 0)$};
          \node[state] (i) [below of=e] {\Large $(i, 0)$};
          \node[state,rectangle] (j) [below left=2.525cm and 2.25cm of f] {\Large $(j, 1)$};
          \node[state,rectangle,fill=red!40] (k1) [below of=f] {\Large $(k, 1)$};
          \node[state,rectangle,fill=red!40] (k2) [below of=g] {\Large $(k, 2)$};
          \node[state,rectangle] (l) [below of=h] {\Large $(l, 0)$};
          \node[state,rectangle] (m) [below of=i] {\Large $(m, 0)$};
          \node[state,accepting] (n) [below of=k2] {\Large $(n, 3)$};
          \node[state,accepting] (o) [below of=l] {\Large $(o, 3)$};
        
          \path 
                (c) edge              node {} (f)
                (c) edge              node {} (g)
                (c) edge              node {} (h)
                (d) edge              node {} (g)
                (d) edge              node {} (h)
                (e) edge              node {} (i)
                (f) edge              node {} (j)
                (f) edge              node {} (k1)
                (g) edge              node {} (k2)
                (h) edge              node {} (l)
                (i) edge              node {} (m)
                (j) edge              node {} (n)
                (k1) edge              node {} (n)
                (k2) edge              node {} (o)
                (l) edge              node {} (n)
                (m) edge              node {} (o)
                (n) edge [loop below]              node {} (n)
                (o) edge [loop below]              node {} (o);
        \end{tikzpicture}
    }\hfill
    \subfloat[\centering Deceptive sure-winning region of P1 when $D = \{l, m\}$\label{fig:decoy-placement-choice-2}]{%
      \begin{tikzpicture}[->,>=stealth',shorten >=1pt,auto,node distance=2.5cm, scale = 0.45,transform shape]

          \node[state,rectangle] (c) [right of=b] {\Large $(c, 0)$};
          \node[state,rectangle] (d) [right of=c] {\Large $(d, 0)$};
          \node[state,rectangle,fill=blue!10] (e) [right of=d] {\Large $(e, 0)$};
          \node[state] (f) [below of=b] {\Large $(f, 1)$};
          \node[state] (g) [below of=c] {\Large $(g, 2)$};
          \node[state,fill=blue!10] (h) [below of=d] {\Large $(h, 0)$};
          \node[state,fill=blue!10] (i) [below of=e] {\Large $(i, 0)$};
          \node[state,rectangle] (j) [below left=2.525cm and 2.25cm of f] {\Large $(j, 1)$};
          \node[state,rectangle] (k1) [below of=f] {\Large $(k, 1)$};
          \node[state,rectangle] (k2) [below of=g] {\Large $(k, 2)$};
          \node[state,rectangle,fill=red!40] (l) [below of=h] {\Large $(l, 0)$};
          \node[state,rectangle,fill=red!40] (m) [below of=i] {\Large $(m, 0)$};
          \node[state,accepting] (n) [below of=k2] {\Large $(n, 3)$};
          \node[state,accepting] (o) [below of=l] {\Large $(o, 3)$};
        
          \path 
                (c) edge              node {} (f)
                (c) edge              node {} (g)
                (c) edge              node {} (h)
                (d) edge              node {} (g)
                (d) edge              node {} (h)
                (e) edge              node {} (i)
                (f) edge              node {} (j)
                (f) edge              node {} (k1)
                (g) edge              node {} (k2)
                (h) edge              node {} (l)
                (i) edge              node {} (m)
                (j) edge              node {} (n)
                (k1) edge              node {} (n)
                (k2) edge              node {} (o)
                (l) edge              node {} (n)
                (m) edge              node {} (o)
                (n) edge [loop below]              node {} (n)
                (o) edge [loop below]              node {} (o);
        \end{tikzpicture}
    }
    
    \caption{Deceptive sure-winning region of P1 under different choice of $D$.}
    \label{fig:decoy-placement-choice}

\end{figure}
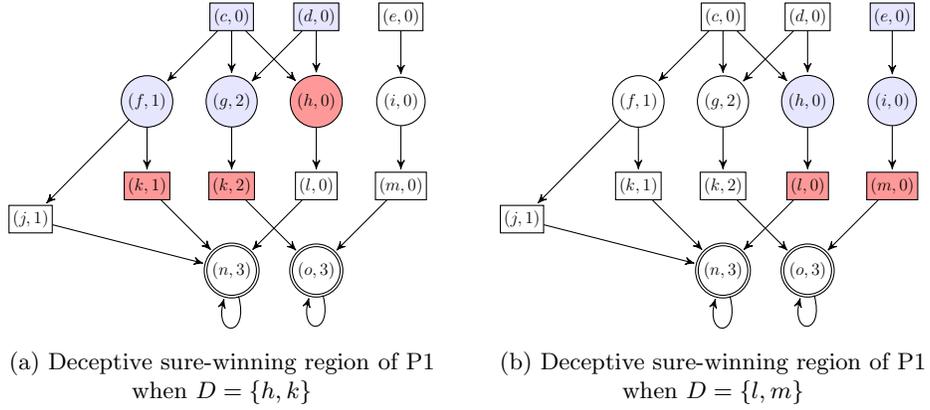
\end{example}

We now define a composition operator  $\biguplus$ over deceptive sure-winning regions which represent the true effect of adding a new state to a given set of decoys. That is, given $D \subseteq \calD$ and $s \in \calD$, let $\biguplus$ be an operator such that 
\[
    \dswin_{D}\biguplus \dswin_{\{s\}}= \dswin_{D\cup \{s\}}.
\] 
 That is,  the composition operator   returns the deceptive sure-winning region in the reachability game     $\langle \win_2^2, \act, \widehat\Delta, \dswin_{D} \cup \dswin_{\{s\}} \rangle$, which equals P1's deceptive sure-winning region   when the set $D \cup \{s\}$ are selected to be decoys.

With this notation, Problem~\ref{prob:decoy-placement} becomes equivalent to identifying a set $D^\ast \subseteq \calD$ such that 

\begin{align}
    \label{eq:composition_opt}
    D^\ast = \arg\max_{D \subseteq \calD} \left| \biguplus\limits_{s \in D} \dswin_{\{s\}} \right| \quad \mbox{subject to: } \card{D}\le k.
\end{align}

It is noted that if we replace the composition operator $\biguplus$ with the union operator $\cup$ in \eqref{eq:composition_opt}, then the problem becomes

\begin{align}
    \label{eq:2}
   \max_{D \subseteq \calD} \left| \bigcup\limits_{s \in D} \dswin_{\{s\}} \right| \quad \mbox{subject to: } \card{D}\le k.
\end{align}
which is a maximum set-cover problem. The maximum set-cover problem is well-known submodular optimization problem and can be solved using a greedy algorithm: Given the current choice $D_i$ of decoys at iteration $i$,  the greedy algorithm  selects a new decoy $s\in \calD\setminus D_i$ that covers the greatest number of uncovered states in $\win_2^2$. This selection iterates until $k$ decoys are selected. It is also known that the greedy algorithm is $(1-1/e)$-approximate. The reader is referred to \cite{vazirani_approximation_2003} for more details.

Let $f^\cup(D) = \left| \bigcup\limits_{s \in D} \dswin_{\{s\}} \right|$ and $f^\uplus(D) = \left| \biguplus\limits_{s \in D} \dswin_{\{s\}} \right|$. It follows from Corollary~\ref{cor:subset} that $f^\cup(D) \leq f^\uplus(D)$ for all $D \subseteq \calD$. In other words, $f^\cup(D)$ under-approximates the effectiveness of allocating the states in $D$ as decoys, which is captured by $f^\uplus(D)$.

While the function $f^\cup$ is submodular, a similar sub/supermodularity condition does not necessarily hold for the function $f^\uplus$. In the sequel, we provide sufficient conditions on when $f^\uplus$ is submodular and when it is supermodular.

\begin{theorem}
    \label{thm:composition}
    The following statements about $f^\uplus(D)= \left| \biguplus\limits_{s \in D} \dswin_{\{s\}} \right|$ are true. 
    \begin{enumerate}
    \item[(a)] $f^\uplus $ is monotone and non-decreasing.
        \item[(b)] $f^\uplus$ is submodular if $\dswin_{D \cup \{s\}} = \dswin_{D} \cup \dswin_{\{s\}}$ for all $D \subseteq \calD$ and $s \in \calD$.
        
        \item[(c)] $f^\uplus$ is supermodular if $\dswin_{D} = \dswin_{D \cup \{s_1\}} \cap \dswin_{D \cup \{s_2\}}$ for all $D \subseteq \calD$ and all $s_1, s_2, \in \calD$.
    \end{enumerate}
\end{theorem}

\begin{proof} 
(\textbf{a}).   Based on Corollary ~\ref{cor:subset}, for any set $D\subseteq \calD$ and a state $s\in \calD\setminus D$, $f^\uplus(D)=\abs{\dswin_{D} } $ and $f^\uplus(D \cup \{s\})  = \abs{\dswin_{D\cup\{s\}} } $, because $\dswin_{D}\subseteq \dswin_{D\cup\{s\}} $, $f^\uplus(D) \le f^\uplus (D\cup \{s\})$.

    (\textbf{b}). When $\dswin_{D \cup \{s\}} = \dswin_{D} \cup \dswin_{\{s\}}$, we can write $f^\uplus(D) = \left| \biguplus \limits_{s \in D} \dswin_{\{s\}} \right| = \left| \bigcup \limits_{s \in D} \dswin_{\{s\}} \right| = f^\cup(D)$, which is submodular. 

   (\textbf{c}).  We will show that 
    \[ LHS\coloneqq f^\uplus(D \cup \{s_1\}) + f^\uplus(D \cup \{s_2\}) - f^\uplus(D) \leq f^\uplus(D \cup \{s_1, s_2\}) \coloneqq RHS\]
    for all $D \subseteq \calD$ and all $s_1, s_2 \in \calD$. 
    Given that $\dswin_{D} = \dswin_{D \cup \{s_1\}} \cap \dswin_{D \cup \{s_2\}}$ holds for any $D \subseteq \calD$ and any $s_1, s_2 \in \calD$, we have that 
    $f^\uplus(D \cup \{s_1\}) + f^\uplus(D \cup \{s_2\}) - f^\uplus(D)$ counts every state in $\dswin_{D\cup\{s_1\}} \cup \dswin_{D\cup\{s_2\}}$ exactly once. On the other hand, we have $f^\uplus(D \cup \{s_1, s_2\}) = \abs{\dswin_{D \cup \{s_1, s_2\}}}$ and $\dswin_{D \cup \{s_1, s_2\}} \supseteq \dswin_{D \cup \{s_1\}} \cup \dswin_{D \cup \{s_2\}}$, by Corollary~\ref{cor:subset}. Thus, there may exist a state in $\dswin_{D \cup \{s_1, s_2\}}$ which is not included in either $\dswin_{D \cup \{s_1\}}$ or $\dswin_{D \cup \{s_2\}}$. In other words, RHS may be greater than or equal to LHS and the statement follows. 
\end{proof}

Based on Thm.~\ref{thm:composition}, we now propose a greedy algorithm  described in Alg.~\ref{alg:greedymax}. This greedy algorithm is an extension of the GreedyMax algorithm for maximizing monotone submodular-supermodular functions in \cite{bai_greed_2018}. It starts with an empty set of states labeled with $\decoy$ and incrementally adds new decoys in the the game arena. At each step, given the deceptive winning region of the chosen decoys,  a new decoy is selected such that by adding the new decoy into the chosen decoys, P1's deceptive sure-winning region covers the largest number of states in $\win_2^2$. The algorithm iterates until $k$ decoys are added, where $k$ is the upper bound on the number of decoys.

\begin{algorithm}
\SetAlgoLined
\KwIn{P1's deceptive reachability game $\langle \win_2^2, \act, \widehat \Delta, F_D=\emptyset\rangle$, the set $\calD\subseteq S$, the bound $k$ on the number of decoys.}
\KwOut{An approximate  solution $\overline D$ for the optimization problem in Eq.~\ref{eq:composition_opt}.}
$\overline D\gets \emptyset$\;
$\dswin_{\overline{D}} \gets \emptyset$\;
 \While{$\card{\overline D}< k$}{
 \For{$s\in \calD\setminus \overline{D}$}{$\game_s \gets \langle \win_2^2, \act, \widehat \Delta, \dswin_{\{s\}}\cup \dswin_{\overline{D}} \rangle $\;
 $\dswin_{\{s\}\cup \overline{D}} \gets \mbox{Sure-Win}(\game_s)$; \hspace{10.5em} \text{... by Alg.\ref{alg:sure-win}}\; 
 }
  $s^\ast \gets \arg\max_{s\in \calD\setminus \overline{D}}\left| \dswin_{\{s\}\cup \overline D}\right|$\;
  $\overline{D}\gets s^\ast \cup \overline{D}$\;
  }
  \Return{$\overline{D}$}
 \caption{GreedyMax Algorithm for Decoy Allocation}
 \label{alg:greedymax}
\end{algorithm}

\paragraph*{Complexity Analysis} The complexity of Alg.~\ref{alg:greedymax} is $\mathcal{O}(k\abs{\calD}N)$ where $N$ is the number of state-action pairs in P1's deceptive reachability game. This is because  to add $(i+1)$-th state to $\overline{D}$, we update deceptive sure-winning regions of $\abs{\calD} - i$ states. The complexity of solving a reachability game is linear in the size $N$ of the game, measured by the number of state-action pairs.

\setcounter{example}{1}
\begin{example}[Part 5]
    \label{ex:running-example-5}
    We maximize $\abs{\dswin_D}$, under the constraint that  a maximal two decoys to be placed within the set $\calD= \{j, k, l, m\}$. Following the compositional approach, we compute the following deceptive sure-winning regions:  $\dswin_{\{j\}} = \{(j, 1), (f, 1)\}$, $\dswin_{\{k\}} = \{(k, 1), (k, 2), (f, 1), (g, 2)\}$, $\dswin_{\{l\}} = \{(l, 0), (h, 0)\}$ and $\dswin_{\{m\}} = \{(m, 0),$ $(i, 0), (e, 0)\}$.

    First, we use the greedy algorithm for maximum set-cover to solve for $D \subseteq \calD$ that maximizes $f^{\cup}(D)$ under the constraint $\card{D}\le 2$. In the first iteration, the greedy algorithm selects the largest the state corresponding to $\abs{\dswin_{\{s\}}}$, which is $s = k$. In the second iteration, it selects the set that has the largest number of states not already included in $\dswin_{\{k\}}$. Thus, it selects $m$ as the second state to place the decoy. In conclusion, it selects $D = \{k, m\}$ as solution to decoy allocation problem, for which $\abs{\dswin_{\{k, m\}}} = 7$.

    Second, we use Alg.~\ref{alg:greedymax} to solve for $D \subseteq \calD$ that maximizes $f^{\uplus}(D)$ under the constraint $\card{D}\le 2$. In the first iteration, $s^\ast$ is selected to be $k$ because $\abs{\dswin_{\{k\}}}$ is the largest. In the second iteration, $s^\ast$ is selected to be $l$ because $\dswin_{\{l\} \cup \overline D} = \{(l, 0), (h, 0), (k, 1), (k, 2), (f, 1), (g, 2), (c, 0), (d, 0)\}$. In conclusion, it selects $D = \{k, l\}$ as solution to decoy allocation problem, for which $\abs{\dswin_{\{k, l\}}} = 8$, which coincidentally in this example is also the globally optimal solution for the problem. We note the improvement in the solution is attributed to incremental computation of $\dswin_{D \cup \{s\}}$ in Alg.~\ref{alg:greedymax}. 
\end{example}   

Due to space limitation, we omit other examples with larger game arena. But the interested readers can find more examples  in which the decoy allocation problems are solved with both the greedy algorithm for submodular optimization and Alg.~\ref{alg:greedymax} in \href{https://github.com/abhibp1993/decoy-allocation-problem}{https://github.com/abhibp1993/decoy-allocation-problem}.

\section{Conclusion}
    \label{sec:conclusion}

In this paper, we investigated the optimal decoy allocation problems in a class of games where players' objectives are specified in temporal logic and players have asymmetric information. The contributions of the paper are twofold: First, we develop a hypergame on graph model to capture the deceivee (the adversary)'s  incomplete and incorrect information due to the decoys and the deceiver (the defender)'s information about the deceivee's information. Using decoy-based deception, we designed algorithms to compute a deceptive sure-winning strategy with which the defender can take actions deceptively and lure the adversary into decoys, from a state where the adversary perceives herself a winner (\ie, has a strategy to achieve the attack objective). Second, to compute the optimal choice of decoy locations, we employed compositional synthesis from formal methods and proved that the optimal decoy allocation problem is   monotone, and non-decreasing. However, the problem can be submodular or supermodular or neither in different games. We design two greedy algorithms, one is based on maximizing an under-approximation of the deceptive winning regions given the effectiveness of individual decoys using maximum set cover, another is to use submodular-supermodular optimization to find approximate solutions of the optimal decoy placement.

Future work  include the study of decoy allocation with other types of decoy-induced misperception. In this scope, the decoys are set up as ``traps'' for the adversary. But it is possible to use decoys as ``fake targets'' for distracting the adversary. We intend to explore a mixture of types of  decoys given their functionalities in cyber-physical defense and the respective deceptive synthesis problems and decoy-allocation problems. Also, we are interested in deceptive planning for other class of games, for example, concurrent(\ie, simultaneous-move) reachability games \cite{de_alfaro_concurrent_2007}.  We intend to implement a toolbox for the proposed algorithm and apply the  methods to practical network security problems.

\bibliographystyle{splncs04}
\bibliography{refs}

\begin{thebibliography}{10}
\providecommand{\url}[1]{\texttt{#1}}
\providecommand{\urlprefix}{URL }
\providecommand{\doi}[1]{https://doi.org/#1}

\bibitem{de_alfaro_concurrent_2007}
de~Alfaro, L., Henzinger, T.A., Kupferman, O.: Concurrent {Reachability}
  {Games}. Theoretical Computer Science  \textbf{386}(3),  188--217 (Nov 2007)

\bibitem{anwar_honeypot_2020}
Anwar, A.H., Kamhoua, C., Leslie, N.: Honeypot {Allocation} over {Attack}
  {Graphs} in {Cyber} {Deception} {Games}. In: 2020 {International}
  {Conference} on {Computing}, {Networking} and {Communications} ({ICNC}). pp.
  502--506 (2020)

\bibitem{aslanyan_quantitative_2016}
Aslanyan, Z., Nielson, F., Parker, D.: Quantitative {Verification} and
  {Synthesis} of {Attack}-{Defence} {Scenarios}. In: 2016 {IEEE} 29th
  {Computer} {Security} {Foundations} {Symposium} ({CSF}). pp. 105--119. IEEE
  (Jun 2016)

\bibitem{bai_greed_2018}
Bai, W., Bilmes, J.A.: Greed is still good: {Maximizing} monotone submodular+
  supermodular functions. arXiv preprint arXiv:1801.07413  (2018)

\bibitem{baier2008principles}
Baier, C., Katoen, J.P.: Principles of model checking (2008)

\bibitem{bennett_hypergame_1986}
Bennett, P.G., Bussel, R.R.: Hypergame {Theory} and {Methodology}: {The}
  {Current} “{State} of the {Art}”. In: Wilkin, L. (ed.) The {Management}
  of {Uncertainty}: {Approaches}, {Methods} and {Applications}, pp. 158--181.
  Springer Netherlands, Dordrecht (1986)

\bibitem{bloem_graph_2018}
Bloem, R., Chatterjee, K., Jobstmann, B.: Graph {Games} and {Reactive}
  {Synthesis}. In: Clarke, E.M., Henzinger, T.A., Veith, H., Bloem, R. (eds.)
  Handbook of {Model} {Checking}, pp. 921--962. Springer International
  Publishing (2018)

\bibitem{chatterjee_survey_2012}
Chatterjee, K., Henzinger, T.A.: A survey of stochastic omega-regular games.
  Journal of Computer and System Sciences  \textbf{78}(2),  394--413 (2012)

\bibitem{durkota_optimal_2015}
Durkota, K., Lisy, V., Bosansky, B., Kiekintveld, C.: Optimal {Network}
  {Security} {Hardening} {Using} {Attack} {Graph} {Games}. In: Twenty-{Fourth}
  {International} {Joint} {Conference} on {Artificial} {Intelligence} (2015)

\bibitem{filiot_antichains_2011}
Filiot, E., Jin, N., Raskin, J.F.: Antichains and {Compositional} {Algorithms}
  for {LTL} {Synthesis}. Formal Methods in System Design pp. 261--296 (Dec
  2011)

\bibitem{gradel_automata_2002}
Gradel, E., Thomas, W.: Automata, {Logics}, and {Infinite} {Games}: {A} {Guide}
  to {Current} {Research}, vol.~2500. Springer Science \& Business Media (2002)

\bibitem{heckman_bridging_2015}
Heckman, K.E., Stech, F.J., Thomas, R.K., Schmoker, B., Tsow, A.W.: Bridging
  the {Classical} {D}\&{D} and {Cyber} {Security} {Domains}. In: Cyber
  {Denial}, {Deception} and {Counter} {Deception}: {A} {Framework} for
  {Supporting} {Active} {Cyber} {Defense}, pp. 5--29. Springer International
  Publishing (2015)

\bibitem{jha_two_2002}
Jha, S., Sheyner, O., Wing, J.: Two {Formal} {Analyses} of {Attack} {Graphs}.
  In: Proceedings 15th {IEEE} {Computer} {Security} {Foundations} {Workshop}.
  {CSFW}-15. pp. 49--63 (2002)

\bibitem{jiang_optimal_2009}
Jiang, W., Fang, B.x., Zhang, H.l., Tian, Z.h., Song, X.f.: Optimal {Network}
  {Security} {Strengthening} {Using} {Attack}-{Defense} {Game} {Model}. In:
  2009 {Sixth} {International} {Conference} on {Information} {Technology}:
  {New} {Generations}. pp. 475--480 (2009)

\bibitem{kiekintveld_computing_2009}
Kiekintveld, C., Jain, M., Tsai, J., Pita, J., Ordóñez, F., Tambe, M.:
  Computing optimal randomized resource allocations for massive security games.
  In: Proceedings of {The} 8th {International} {Conference} on {Autonomous}
  {Agents} and {Multiagent} {Systems}-{Volume} 1. pp. 689--696 (2009)

\bibitem{kiekintveld_game-theoretic_2015}
Kiekintveld, C., Lisỳ, V., Píbil, R.: Game-theoretic foundations for the
  strategic use of honeypots in network security. In: Cyber {Warfare}, pp.
  81--101. Springer (2015)

\bibitem{kulkarni_deceptive_2021}
Kulkarni, A.N., Luo, H., Leslie, N.O., Kamhoua, C.A., Fu, J.: Deceptive
  {Labeling}: {Hypergames} on {Graphs} for {Stealthy} {Deception}. IEEE Control
  Systems Letters  \textbf{5}(3),  977--982 (2021)

\bibitem{kulkarni_compositional_2018}
Kulkarni, A.N., Fu, J.: A compositional approach to reactive games under
  temporal logic specifications. In: 2018 {Annual} {American} {Control}
  {Conference} ({ACC}). pp. 2356--2362. IEEE (2018)

\bibitem{manna_temporal_1992}
Manna, Z., Pnueli, A.: The {Temporal} {Logic} of {Reactive} and {Concurrent}
  {Systems}: {Specification}. Springer-Verlag (1992)

\bibitem{mcnaughton_infinite_1993}
McNaughton, R.: Infinite games played on finite graphs. Annals of Pure and
  Applied Logic  \textbf{65}(2),  149--184 (1993)

\bibitem{ouScalableApproachAttack2006}
Ou, X., Boyer, W.F., McQueen, M.A.: A scalable approach to attack graph
  generation. In: Proceedings of the 13th {{ACM}} Conference on {{Computer}}
  and Communications Security - {{CCS}} '06. pp. 336--345. {ACM Press},
  {Alexandria, Virginia, USA} (2006)

\bibitem{pibil_game_2012}
Píbil, R., Lisỳ, V., Kiekintveld, C., Bošanskỳ, B., Pěchouček, M.: Game
  theoretic model of strategic honeypot selection in computer networks. In:
  International {Conference} on {Decision} and {Game} {Theory} for {Security}.
  pp. 201--220. Springer (2012)

\bibitem{sasaki_hierarchical_2016}
Sasaki, Y., Kijima, K.: Hierarchical {Hypergames} and {Bayesian} {Games}: {A}
  {Generalization} of the {Theoretical} {Comparison} of {Hypergames} and
  {Bayesian} {Games} {Considering} {Hierarchy} of {Perceptions}. Journal of
  Systems Science and Complexity  \textbf{29}(1),  187--201 (Feb 2016)

\bibitem{sinha_stackelberg_2018}
Sinha, A., Fang, F., An, B., Kiekintveld, C., Tambe, M.: Stackelberg {Security}
  {Games}: {Looking} {Beyond} a {Decade} of {Success}. In: Proceedings of the
  {Twenty}-{Seventh} {International} {Joint} {Conference} on {Artificial}
  {Intelligence}. pp. 5494--5501. International Joint Conferences on Artificial
  Intelligence Organization (2018)

\bibitem{thakoor_cyber_2019}
Thakoor, O., Tambe, M., Vayanos, P., Xu, H., Kiekintveld, C., Fang, F.: Cyber
  {Camouflage} {Games} for {Strategic} {Deception}. In: Alpcan, T.,
  Vorobeychik, Y., Baras, J.S., Dán, G. (eds.) Decision and {Game} {Theory}
  for {Security}. pp. 525--541. Lecture {Notes} in {Computer} {Science},
  Springer International Publishing (2019)

\bibitem{vazirani_approximation_2003}
Vazirani, V.V.: Approximation {Algorithms}. Springer-Verlag (2003)

\bibitem{wang_solution_1989}
Wang, M., Hipel, K.W., Fraser, N.M.: Solution {Concepts} in {Hypergames}.
  Applied Mathematics and Computation  \textbf{34}(3),  147--171 (Dec 1989)

\end{thebibliography}

\end{document}